\documentclass[12pt]{article}
\usepackage{natbib,graphicx}
\usepackage{amsmath,amsthm,amssymb,enumerate,enumitem}
\usepackage{thmtools} 
\usepackage{multirow,array,url}
\usepackage{authblk}
\usepackage[nodisplayskipstretch]{setspace}
\usepackage{titlesec}
\usepackage[margin=1in]{geometry}

\newcommand{\E}{\mathbb{E}}

\newcommand{\R}{\mathbb{R}}
\newcommand{\Alpha}{\mathrm{A}}
\renewcommand{\P}{\mathbb{P}}

\newcommand{\Cov}{\mathrm{Cov}}
\newcommand{\diag}{\mathrm{diag}}

\setlist[1]{itemsep=-7pt}

\titlespacing*{\section}
  {0pt}{0\baselineskip}{0\baselineskip}
\titlespacing*{\subsection}
  {0pt}{0\baselineskip}{0\baselineskip}
\newtheorem{theorem}{Theorem}[section]

\declaretheorem[style=definition,numberwithin=section]{procedure}
\begin{document}

\title{Robust Spatial Extent Inference with a Semiparametric Bootstrap Joint Testing Procedure}

\author[1]{Simon N. Vandekar}
\author[2]{Theodore D. Satterthwaite}
\author[2]{Cedric H. Xia}
\author[2]{Kosha Ruparel}
\author[2]{Ruben C. Gur}
\author[2]{Raquel E. Gur}
\author[3]{Russell T. Shinohara}

\affil[1]{Vanderbilt University, Department of Biostatistics}
\affil[2]{University of Pennsylvania School of Medicine, Department of Psychiatry}
\affil[3]{University of Pennsylvania School of Medicine, Department of Biostatistics, Epidemiology, and Informatics}

\singlespace
\maketitle

{\bf Abstract:}
Spatial extent inference (SEI) is widely used across neuroimaging modalities to study brain-phenotype associations that inform our understanding of disease.
Recent studies have shown that Gaussian random field (GRF) based tools can have inflated family-wise error rates (FWERs).
This has led to fervent discussion as to which preprocessing steps are necessary to control the FWER using GRF-based SEI.
The failure of GRF-based methods is due to unrealistic assumptions about the covariance function of the imaging data.
The permutation procedure is the most robust SEI tool because it estimates the covariance function from the imaging data.
However, the permutation procedure can fail because its assumption of exchangeability is violated in many imaging modalities.
Here, we propose the (semi-) parametric bootstrap joint (PBJ; sPBJ) testing procedures that are designed for SEI of multilevel imaging data.
The sPBJ procedure uses a robust estimate of the covariance function, which yields consistent estimates of standard errors, even if the covariance model is misspecified.
We use our methods to study the association between performance and executive functioning in a working fMRI study.
The sPBJ procedure is robust to variance misspecification and maintains nominal FWER in small samples, in contrast to the GRF methods.
The sPBJ also has equal or superior power to the PBJ and permutation procedures.
We provide an {\tt R} package \url{https://github.com/simonvandekar/pbj} to perform inference using the PBJ and sPBJ procedures.

~\\
Please address correspondence to: \\
Simon N. Vandekar\\
2525 West End Ave., \#1136\\
Department of Biostatistics\\
Vanderbilt University\\
Nashville, TN 37203\\
{\tt simon.vandekar@vanderbilt.edu}




\section{Introduction}

Investigating neuroimage-phenotype associations is important for understanding how stimuli, environment, and diseases affect the brain.
Various neuroimaging modalities are used to study brain-behavior associations. For example, imaging scientists can study how brain activation measured with functional magnetic resonance imaging (fMRI) is related to performance on an executive functioning task \citep{satterthwaite_functional_2013}.
Spatial extent inference (SEI) is a widely used tool across neuroimaging modalities to study image-phenotype associations.
SEI views the statistical image as a realization of a stochastic process and thresholds the image at a given value in order to compute the distribution of the spatial extent of contiguous clusters of activation.
A p-value is determined for each cluster based on its spatial extent by comparing the cluster extent to the distribution of the maximum cluster size under a predetermined null hypothesis \citep{friston_assessing_1994,worsley_detecting_1999}.

A number of recent studies have shown that the standard tools used for SEI that rely on gaussian random field (GRF) approximations can have highly inflated false positive rates \citep{eklund_cluster_2016,silver_false_2011,greve_false_2018}.
These studies have provoked significant discussion about the conditions necessary for the GRF-based tools to produce nominal family-wise error rates \citep[FWERs;][]{kessler_reevaluating_2017,mueller_commentary:_2017,eklund_cluster_2018,slotnick_cluster_2017}.
A good deal of this discussion has led to ad hoc preprocessing and thresholding suggestions for particular data types that are based on simulation results \citep{eklund_cluster_2016,greve_false_2018,flandin_analysis_2016,slotnick_cluster_2017,mueller_commentary:_2017}.
For example, simulation results suggest using particular values preprocessing parameters such that the GRF based methods yield nominal FWERs \citep{flandin_analysis_2016}.
However, GRF-based methods cannot control the FWER generally due to unrealistic assumptions about the covariance structure of the neuroimaging data.
As a result ad hoc preprocessing decisions are unlikely to yield nominal FWERs across a wide range of data sets.

To date, only permutation testing procedures have been shown to robustly control the spatial extent false positive rate in most neuroimaging datasets \citep{eklund_cluster_2016,winkler_permutation_2014}.
However, the permutation procedure may not control the false positive rate at the nominal level in some cases because it relies on the assumption of exchangeability.
Exchangeability includes the assumption of homoskedasticity; namely, that the covariance structures of all the subjects are unaffected by the covariates.
This assumption is often violated in neuroimaging datasets because many types of imaging datasets require multilevel models, in which parameters are estimated separately for each subject and subsequently analyzed at the group level.
Exchangeability may also be violated because the covariance structure of imaging data is known to be affected by in-scanner motion, age, and diagnostic status \citep[see e.g.][]{satterthwaite_impact_2012,power_development_2010,woodward_mapping_2016}.


In the statistics literature several procedures have been developed for hypothesis testing in spatial datasets.
\citet{sun_false_2015} developed false discovery rate (FDR) controlling procedures in a Bayesian context in order to perform inference pointwise or on predefined clusters.
Procedures to address large-scale hypothesis testing for correlated tests could also be utilized to perform inference on images \citep[e.g.][]{yekutieli_resampling-based_1999,benjamini_control_2001,romano_control_2008}.
These procedures, which perform inference pointwise have the desirable property that they identify locations of the image or map where the mean is likely to be non-zero.
However, due to the high-dimensionality, small sample sizes, and complex covariance functions of neuroimaging data, many studies lack power to perform inference directly on the mean function.
As a solution, other proposed methods for controlling the FDR in spatial data improve power by using predefined clusters and performing inference on the clusters \citep{benjamini_false_2007,pacifico_false_2004}.
SEI is similar to these cluster based approaches because a significant cluster only guarantees (with a prespecified probability) that an arbitrary nonzero portion of the cluster has nonzero mean.
However, unlike the cluster based FDR methods SEI does not require that the user defines the clusters {\it a priori}.

Here we develop two methods to perform SEI that use a (semi-) parametric bootstrap joint (PBJ; sPBJ) testing procedure to compute test statistics and estimate the joint distribution of test statistics (Section \ref{sec:theory}).
The PBJ procedure accomodates violations of exchangeability by using subject weights to deweight noisy observations at each voxel.
If the selected weights are inversely proportional to the subject-level variances then the covariance function is correctly specified and the parameter estimators have minimum variance.
In order to address the case that subject weights are not correctly specified we introduce the sPBJ procedure that uses a robust ``sandwich" estimate of the covariance function, which is an asymptotically consistent estimator of the covariance function of the stochastic process \citep[estimator $\text{HC}_3$ in][]{long_using_2000,mackinnon_heteroskedasticity-consistent_1985}.
Because the covariance function of the stochastic process is complex in imaging data, we leverage a bootstrap procedure to estimate the distribution of the maximum cluster size \citep{vandekar_faster_2017}.
These methods are asymptotically robust across any range of preprocessing parameters and cluster forming thresholds (CFTs) because they estimate the covariance structure from the observed data.

We utilize a subset of 1,000 subjects from the Philadelphia Neurodevelopmental Cohort \citep[PNC;][]{satterthwaite_neuroimaging_2014,satterthwaite_philadelphia_2016} and a novel approach to generate realistic simulated data in which the assumption of exchangeability is violated in order to show that the PBJ and sPBJ procedures are robust to violations of this assumption (Section \ref{sec:simulation}).
We demonstrate the procedures by studying the association between task performance and executive functioning measured with fMRI using the $N$-back task \citep[Sections \ref{sec:nback} and \ref{sec:nbackanalysis};][]{ragland_working_2002,satterthwaite_functional_2013}.
While the procedure is generally applicable for inference in spatial data sets, we developed an {\tt R} package to perform the PBJ and sPBJ SEI procedures on NIfTI images that is available for download at \url{https://github.com/simonvandekar/pbj} and from Neuroconductor \citep[\url{https://www.neuroconductor.org/};][]{muschelli_neuroconductor:_2018}.

\section{Spatial extent inference}

Let the statistical map under the null, $Z_0(v)$, be a stochastic process on $\mathcal{L}^\infty(\mathbb{B},\R)$, where $\mathcal{L}^\infty(\mathbb{B}, \R)$ denotes the normed space of bounded functions from a bounded subset $\mathbb{B} \subset \R^3$ to $\R$ equipped with the uniform norm.
Let $z(v) \in \mathcal{L}^\infty(\mathbb{B},\R)$ be the observed statistical map that is sampled from an arbitrary stochastic process $Z(v)$ on $\mathcal{L}^\infty(\mathbb{B},\R)$.
We assume that the sample paths of $Z_0(v)$ and $Z(v)$ are continuous almost everywhere.
Let $\mathcal{C}(z(v), z_0) : \mathcal{L}^\infty(\mathbb{B},\R) \times \R \to \R^\infty$ be a function that thresholds the map $\lvert z(v)\rvert$ with a given cluster forming threshold (CFT), $z_0 \in \R$, and computes the volume of all suprathreshold contiguous clusters in the image.
$\mathcal{C}(z(v), z_0)$ is in $\R^\infty$ because there is an arbitrary number of contiguous clusters possible.
To give the clusters a natural ordering, we assume that, for all sample paths $z(v)$, $\mathcal{C}(z(v), z_0)_{j} \ge \mathcal{C}(z(v), z_0)_{k}$ for $j<k$ and that there exists finite $J$ such that $ \mathcal{C}(z(v), z_0)_j = 0$ for all $j\ge J$.
That is, the clusters are indexed in decreasing order and there are a finite number of nonzero clusters.

\begin{procedure}[Spatial extent inference]
\label{proc:seht}
For a fixed $z_0 \in \R$ we set
\begin{equation}
\label{eq:pvalue}
p_{j} = \P\left\{ \mathcal{C}(Z_0(v), z_0)_1 \ge \mathcal{C}(z(v), z_0)_j \right\}
\end{equation}
and call cluster $j$ ``significant" if $p_j< \alpha$ for some predetermined rejection threshold $\alpha \in [0,1]$.
\end{procedure}

The p-value \eqref{eq:pvalue} computed by Procedure \ref{proc:seht} represents the probability of observing a cluster as or more extreme than $\mathcal{C}(z(v), z_0)_j$, under the global null that $z$ is a sample from $Z_0$ \citep{friston_assessing_1994}.
Typically, the threshold $z_0$ is chosen such that suprathreshold regions have uncorrected p-values less than some common CFT threshold, e.g. $z_0 = 2.32$, which corresponds to a voxel-wise $p<0.01$ under Gaussianity assumptions.
The bootstrap samples generated by Procedures \ref{proc:PBJbootstrap} and \ref{proc:sPBJbootstrap}, below, can be used to perform the SEI described in Procedure \ref{proc:seht} by approximating the distribution of $\mathcal{C}(Z_0(v), z_0)_1$ across $B$ bootstrap samples.

\section{Statistical theory for the bootstrap procedures}
\label{sec:theory}

\subsection{Overview}
We propose two approaches to model the covariance function of the statistical map.
The PBJ procedure assumes the within subject variance estimates are correct.
The sPBJ procedure utilizes a working covariance function.
An advantage of the sPBJ is that even if the working covariance function is incorrect the procedure still yields consistent standard errors.

Let $Y_i(v)$, for $i=1,\ldots, n$, be stochastic processes on $\mathcal{L}^\infty(\mathbb{B},\R)$.
$Y_i(v)$ represent the imaging data and $\mathbb{B}$ is the bounded space of interest (in neuroimaging, the brain).
We assume that
\begin{align}\label{eq:model}
\begin{split}
Y_i(v) & = X_{i0} \alpha(v) + X_{i1} \beta(v) + E_i(v)\\
& = X_i \zeta(v) + E_i(v),
\end{split}
\end{align}
where $X_{i0} \in \R^{m_0}$ is a row vector of nuisance covariates including the intercept, $X_{i1}\in \R^{m_1}$ is a row vector of variables of interest, $m=m_0 + m_1$, $X_i = [ X_{i0}, X_{i1}]$, the parameter image column vectors $\alpha(v) \in \mathcal{L}^\infty(\mathbb{B}, \R^{m_0})$, $\beta(v) \in \mathcal{L}^\infty(\mathbb{B}, \R^{m_1})$, $\zeta(v) = [\alpha(v)^T, \beta(v)^T ]^T$, error $E_i(v)$ with $\E\{ E_i(v)\} = 0$ and covariance function $\Sigma_i(v, w)  = \Cov\left\{E_i(v), E_i(w)\right\} < \infty$ for all $v,w \in \mathbb{B}$.
Here, all capital letters denote random variables or stochastic processes.
Throughout, we assume that the conditional mean $\E\{ Y_i(v) \mid X_i\} = X_i \zeta(v)$ of model \eqref{eq:model} is correctly specified.
This model is a form of function-on-scalar regression \citep{ramsay_functional_2005,reiss_fast_2010,morris_functional_2015} and encompasses many models of interest in image analysis.
Let $X_0 = (X_{10}^T, \ldots, X_{n0}^T) \in \R^{n\times m_0}$, $X_1 = (X_{11}^T, \ldots, X_{n1}^T) \in \R^{n\times m_1}$, $X = [X_0, X_1]$, $Y(v) = (Y_1(v), \ldots, Y_n(v))$.
Lastly, to describe sampling procedures we define $V$ to be the number of locations in the image.


\subsection{PBJ: parametric covariance functions}

Without weights, the basic PBJ model assumes
\begin{equation}
\label{eq:PBJassumption}
\Sigma_i(v,w) = \Sigma(v,w)
\end{equation}
for all $i=1,\ldots,n$; i.e. every subject has the same image covariance function.
This can be generalized to heteroskedastic variances by using weighted regression (with subject specific weights) where the weights are, ideally, proportional to the inverse of the standard deviation, $\sqrt{\Sigma_i(v,v)}$.
Inference using the weighted procedure is conditional on the weights, which are usually estimates (or approximations) of the standard deviation of $E_i(v)$.
This approach is akin to similar multilevel approaches used in neuroimaging software \citep{woolrich_bayesian_2009,penny_random_2007}.
If the weights are incorrectly specified, then the standard error estimates are asymptotically biased and the FWER may not be controlled at the nominal level.
For simplicity, we describe the unweighted version of the PBJ procedure here.
In Section \ref{sec:sPBJ} we describe the robust covariance function estimator used by sPBJ that allows subject specific weights.

Under assumption \eqref{eq:PBJassumption} we can compute the $\chi^2_{m_1}$-statistic image
\begin{equation}
\label{eq:chisquarestat}
Z_0(v) = \Phi_{m_1}\left[\Phi^{-1}_F\left\{T(v) \right\}\right],
\end{equation}
where $T(v)$ is the F-statistic image for the test of the hypothesis
\begin{equation}
\label{eq:null}
H_0(v): \beta(v) = 0,
\end{equation}
$\Phi_F^{-1}$ denotes the inverse cumulative distribution function (CDF) of an F-distribution on $m_1$ and $n-m$ degrees of freedom, and $\Phi_{m_1}$ is the CDF of a $\chi^2_{m_1}$ statistic \citep{vandekar_faster_2017}.
The transformation in \eqref{eq:chisquarestat} makes $Z_0(v)$ asymptotically a diagonal Wishart process \citep{vandekar_faster_2017},
\begin{equation}
\label{eq:PBJdiagwishart}
Z_0 \sim \diag\left\{ \mathcal{W}(m_1, \Sigma_Z) \right\}
\end{equation}
under the null \eqref{eq:null}, where $\Sigma(v,w) = \text{Cor}\left\{E_i(v), E_i(w)\right\}$.
The transformation on the right hand side of \eqref{eq:chisquarestat} improves the finite sample performance of the following procedure, which samples observations from a diagonal Wishart process, conditional on an estimate of $\Sigma_Z$.

\begin{procedure}
\label{proc:PBJbootstrap}
\begin{enumerate}
\item Regress the observed image vector $Y(v)$ onto $X$ as in model \eqref{eq:model} and obtain the residual image
\[
R(v) =
\begin{bmatrix}
R(1) \\
R(2) \\
\vdots \\
R(V)
\end{bmatrix} \in \R^{V,\times n}.
\]
\item Let $\hat\Sigma^{1/2}_Z(v)$ be $R(v)$ after the rows are standardized to be norm 1.
\item For $b=1, \ldots, B$ sample $n\times m_1$ standard normal random variables, $Z_{n,b} \in \R^{n \times m_1}$ and compute the $\diag\left\{ \mathcal{W}(m_1, \hat \Sigma_Z) \right\}$-statistic image
\[
Z_b(v) = \text{diag}\left\{\hat\Sigma^{1/2}_Z(v) \times Z_{n,b} Z_{n,b}^T \times \hat\Sigma^{1/2}_Z(v)^T\right\}.
\]
\end{enumerate}

\end{procedure}

The maximum cluster size of the bootstrap images $Z_b(v)$ are then used to compute the cluster $p$-values \eqref{eq:pvalue}.

\subsection{sPBJ: sandwich covariance functions}
\label{sec:sPBJ}

Here, we utilize an estimating equations approach to obtain parameter estimates and robust standard errors.
Standard error estimates using estimating equations are robust to misspecification of the covariance function which can occur when the subject variance function $\Sigma_i(v, w)$ differs across $i$.
We use the following estimating equation
\begin{equation}\label{eq:ee}
\Psi(\zeta,v;Y) = \sum_{i=1}^n S^{-1}_i(v)(Y_i(v) - X_i\zeta(v)) X_i^T = 0,
\end{equation}
where $S_i(v)$ is a weight for subject $i$ that is an estimator for $\E \left\{ \Sigma_i(v,v) \mid X_i \right\}$ for $v\in \mathbb{B}$.
$S_i(v)$ does not need to be an estimator of $\E \left\{ \Sigma_i(v,v) \mid X_i \right\}$ for the sPBJ procedure to be asymptotically valid, however the best choices are estimators that are nearly unbiased for $\Sigma_i$ and have small variance.
Because our outcome data are stochastic processes indexed by $v \in \mathbb{B}$, \eqref{eq:ee} is also a stochastic process.

We define $\hat \zeta(v)$ as the solution to \eqref{eq:ee}, which is equivalent to the least squares estimate for $\zeta(v)$ at each location $v$ in the image using the weight function $S^{-1}_i(v)$.
$\hat \zeta(v)$ is unbiased in finite samples, if the mean function is correctly specified, and has covariance function equal to 
\begin{align}\label{eq:robustvariance}
\begin{split}
\Sigma_\zeta(v, w)
& = \Cov\left[n^{1/2}\{ \hat \zeta(v) - \zeta(v)\}, n^{1/2}\{ \hat \zeta(w) - \zeta(w)\} \right] \\
& = \Alpha(v)^{-1} \Omega(\zeta, v, w) \Alpha(w)^{-1},
\end{split}
\end{align}
where
\begin{align}\label{eq:AlphaOmega}
\begin{split}
\Alpha(v) & = \lim_{n\to\infty} n^{-1}\sum_{i=1}^n\E\left[ S^{-1}_i(v) X_i^TX_i \right]\\
\Omega(\zeta, v, w) & = \lim_{n\to\infty} n^{-1}\sum_{i=1}^n\E \left[ \{S^{-1}_i(v)S^{-1}_i(w)\} \{Y_i(v) - X_i\zeta(v)\}\{Y_i(w) - X_i\zeta(w)\} X_i^TX_i\right],
\end{split}
\end{align}
and we assume both limits \eqref{eq:AlphaOmega} exist.
We define \eqref{eq:AlphaOmega} as limits because the underlying observations are not identically distributed due to having different subject-level variances.
The covariance \eqref{eq:robustvariance} is derived by performing a Taylor expansion of the estimating equation \eqref{eq:ee} centered at the true value of the parameter $\zeta(v)$ and then computing their asymptotic covariances \citep[see e.g.][p. 300]{boos_essential_2013}.
The following theorem gives a useful way to compute an estimate of the sandwich covariance function for the parameter of interest.
We use a modification (by \citet{long_using_2000}) of the covariance estimate of \citet{mackinnon_heteroskedasticity-consistent_1985}, which has shown to be best at achieving nominal type 1 error rates in small samples over several heteroskedasticity-consistent estimators \citep{long_using_2000}.
\begin{theorem}
\label{thm:betacovariance}
The covariance function $\Sigma_\beta : \mathbb{B} \times \mathbb{B} \to \R^{m_1 \times m_1}$ for the parameter of interest $n^{1/2}(\hat\beta-\beta)$ is
\begin{align*}
\Sigma_\beta(v, w)
= & \Alpha_\beta(v)^{-1} \Omega_\beta(v,w) \Alpha_\beta(w)^{-1}
\end{align*}
where $\Alpha_\beta(v)$ and $\Omega_\beta(v,w)$ are as defined in Supplementary Section \ref{sec:notation}.

The estimator
\begin{align}
\label{eq:sigma}
\hat{\Sigma}_{\beta}(v,w) = & \hat{\Alpha}_\beta(v)^{-1} X_1^T S^{-1/2}(v) P^{X_0}(v) Q(v) Q(w) P^{X_0}(w) S^{-1/2}(w) X_1 \hat{\Alpha}_\beta(w)^{-1},
\end{align}
is consistent for $\Sigma_\beta(v,w)$, where all objects are as defined in Supplementary Section \ref{sec:notation}.

\end{theorem}

The proof of Theorem \ref{thm:betacovariance} is given in Supplementary Section \ref{sec:proofs}.

The Wald statistic variable, $Z(v)$, for the test of the hypothesis image
\begin{equation}\label{eq:H0}
H_0(v): \beta(v) = 0,
\end{equation}
 where $v$ is fixed, is asymptotically chi-squared (under the null),
\begin{equation}\label{eq:teststatistic}
Z(v) = n \hat \beta(v)^T \hat\Sigma_\beta(v)^{-1} \hat \beta(v) \sim \chi^2_{m_1}
\end{equation}
by Theorem 8.3 of \citet[p. 345]{boos_essential_2013}.

\begin{theorem}
\label{thm:m1distribution}
When $m_1=1$ the statistical image $Z(v)$ from \eqref{eq:teststatistic} for the test of \eqref{eq:H0} is asymptotically a diagonal Wishart process,
\begin{equation}\label{eq:wishart}
Z(v) \sim \diag\left\{\mathcal{W}(1; \Sigma )\right\},
\end{equation}
where
\begin{equation*}
\Sigma(v,w) = \Omega_\beta^{-1/2}(v,v) \Omega_\beta(v,w) \Omega_\beta^{-1/2}(w,w).
\end{equation*}
\end{theorem}

The proof of Theorem \ref{thm:m1distribution} is also given in the Appendix.
When $m_1>1$ the joint distribution of $Z(v)$ is challenging to sample from, so that case is beyond the scope of this work.

The following procedure describes how to sample from \eqref{eq:wishart}, conditional on the estimate \eqref{eq:sigma}.

\begin{procedure}\label{proc:sPBJbootstrap}
Let all objects be as defined in Section \ref{sec:notation}. In order to generate samples $Z_{V,b}$, for $b=1,\ldots, B$ from a $\diag\left\{\mathcal{W}(1, \hat \Sigma)\right\}$:
\begin{enumerate}
\item Perform the weighted regression of $Y(v)$ onto $X$ using weights $S^{-1/2}(v)$ to obtain the weighted residuals $R(v)$.
\item Divide the residual vector $R(v)$ elementwise by the diagonal of $P^{X}(v)$ to obtain $Q(v)$.

\item Construct the intermediate $V \times n$ matrix
\[
\hat\Sigma^{1/2}_* =
\begin{bmatrix}
X_1^T S^{-1/2}(1) P_{X_0}(1) \diag\{Q(1)\} \\
X_1^T S^{-1/2}(2) P_{X_0}(2) \diag\{Q(2)\} \\
\vdots \\
X_1^T S^{-1/2}(V) P_{X_0}(V) \diag\{Q(V)\}
\end{bmatrix}
\]
and let $\hat \Sigma^{1/2}$, be $\hat \Sigma^{1/2}_*$ after the rows are standardized to be norm 1.
\item For $b=1, \ldots, B$ sample $n$ independent standard normal random variables, $Z_{n,b}$, and compute the asymptotically normal test statistic vector
\[
Z_{V,b} = \hat \Sigma^{1/2} Z_{n,b}.
\]
\end{enumerate}
\end{procedure}

Recall that Procedures \ref{proc:PBJbootstrap} and \ref{proc:sPBJbootstrap} are used to estimate the probability distribution necessary to use Procedure \ref{proc:seht}.

\section{Fractal $N$-back fMRI data description}
\label{sec:nback}

We use the $N$-back task fMRI data from the PNC to generate realistic data for the simulation analyses.
The fractal $N$-back task is a working memory task where complex geometric fractals were presented to the subject for 500 ms, followed by a fixed interstimulus interval of 2500 ms \citep{ragland_working_2002}. For each subject, approximately 11.5 minutes of $n$-back task fMRI were acquired as BOLD-weighted timeseries (231 volumes; TR = 3000ms; TE = 32ms; FoV = 192x192mm; resolution 3mm isotropic); detailed scanning parameters are discussed in \citet{satterthwaite_neuroimaging_2014,satterthwaite_functional_2013}.
Images were presented under three conditions: In the 0-back condition, participants responded with a button press to a specified target.
For the 1-back condition, participants responded if the current fractal was identical to the previous one; in the 2-back condition, participants responded if the current fractal was identical to the item presented two trials previously \citep{satterthwaite_neuroimaging_2014}.
We include in our analyses a subset of 1,100 subjects who satisfied image quality and clinical exclusionary criteria from the total of 1,601 scanned as part of the PNC.

We apply a series of standard image preprocessing steps: distortion-correction using FSL's FUGUE, time-series preprocessing, rigid registration, brain extraction, temporal filtering, and spatial smoothing \citep{smith_susannew_1997,jenkinson_fsl_2012}.
subject-level models are fit using a linear model in FSL's FILM software including convolved task boxcar functions for each level of the task and 24 motion covariates.
The 2-0 back contrast of these estimates is used as the subject-level outcome for all analyses presented in this paper.
The subject-level 2-0 back contrast's variance is also utilized as a voxel level weight in Section \ref{sec:nbackanalysis}.

The result of the subject-level preprocessing is a parameter estimate image $Y_i(v)$ for each subject, $i=1, \ldots, n$, and a variance estimate $S_i(v)$ for each subject's parameter estimate image.
The parameter estimate images are then submitted to a second level analysis where the parameter estimates are regressed onto covariates (e.g. performance, age, and motion) in order to understand the association between task activation and performance at each voxel across subjects.

\section{Simulation Analysis}
\label{sec:simulation}

\subsection{Methods}

In order to generate nonexchangeable data with heteroskedastic covariance functions we first residualize the $N$-back neuroimaging data set to a set of independent variables that have known associations with the imaging variable, including age, sex, in-scanner-motion (mot), and task performance \citep[$d'$;][]{macmillan_signal_2002}, using the unpenalized spline model
\begin{equation}
\label{eq:nbackmodel}
Y_i(v) = \alpha_0(v) + \alpha_1(v) \times \text{sex}_i + f(v, \text{age}_i) + g(v, \text{mot}_i) + h(v, d'_i) + \epsilon_i(v),
\end{equation}
where $\text{sex}_i$ is an indicator for males, and $f$, $g$, and $h$ are fit using thin plate splines with 10 knots \citep{wood_gams_2002,wood_package_2015}.
Fixed degree spline bases were used to model continuous variables to ensure that the mean of the residuals was completely unassociated with the covariates.

As a result of fitting this model, the residualized images have little to no mean association with the independent variables.
However, the residualized images do not remove the effect the independent variables can have on the covariance structure of the imaging data.
By drawing subsamples of subjects from the residualized images together with the independent variables we generate data that constitutes the global null hypothesis because the covariates are unassociated with the mean, however these covariates can still have an effect on the covariance function of the data.
To demonstrate how the covariates affect the covariance function of the residualized data we create covariance matrices using nodes defined in the Power atlas for the 100 lowest and 100 highest valued subjects for motion and $d'$ \citep[Figure \ref{fig:networks};][]{power_functional_2011}.
Fox example, high motion is associated with stronger interhemispheric and anterior-posterior correlations than low motion (Figure \ref{fig:networks}).

\begin{figure}
\center
\includegraphics[width=0.9\linewidth]{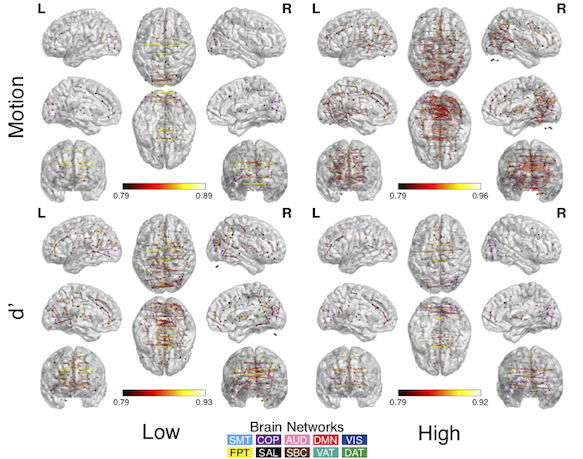}
\caption{Correlation weights for the 100 most extreme observations of motion and $d'$. High motion and low $d'$ are each associated with stronger inter-hemisphere and anterior-posterior correlations. Sphere color indicates network membership. The networks are thresholded for visualization.}
\label{fig:networks}
\end{figure}

In each of 1,000 simulations we draw a bootstrap sample of size $n\in \{25, 50, 200\}$ and fit the model
\begin{equation}
\label{eq:simmodel}
\epsilon_i(v) = \beta_0(v) + \beta_1(v) \times \text{sex}_i + \beta_2(v)\times \text{age}_i + \beta_3(v) \times\text{mot}_i + \beta_4(v)\times d' + \epsilon_{ib}(v),
\end{equation}
where $\epsilon_i(v)$ are the residuals from model \eqref{eq:nbackmodel}.
We then perform the test of $\beta_{j}(v) = 0$ for $j\in \{ 3, 4\}$ using SEI of a $\chi^2_1$-statistical map with two CFTs  ($p<0.01, 0.005$; $\chi^2_1 > 6.63, 7.88$).

We compare the PBJ and sPBJ (using 500 bootstraps) procedures to a standard GRF-based method \citep[easythresh;][]{jenkinson_fsl_2012} and to the permutation procedure \citep[using 500 permutations;][]{winkler_permutation_2014}.
For the simulation study we use the PBJ and sPBJ with uniform weights (PBJ(1); sPBJ(1)) and with the inverse of motion as the weight (PBJ(mot); sPBJ(mot)), which down-weights subjects with high motion.
The permutation procedure is fit using a single group variance estimate (Perm) as well as a two group variance (Perm Grp), which attempts to weaken the assumption of exchangeability by allowing for group differences in variance at each location \citep{winkler_permutation_2014}.
The groups were determined by the upper and lower quantiles of motion.
The SEI FWER is estimated by computing the proportion of simulations where a cluster is determined significant at two rejection thresholds $\alpha=0.01$ and $\alpha=0.05$.

We use a novel procedure to assess power: for each of 1,000 simulations we randomly select 4 voxels within the brain and create an image of spheres, $s(v)$, with radii of 8, 10, 11, and 12 voxels centered at each of the 4 voxels.
We set $s(v)=1$ when $v$ is inside a sphere (and inside the brain) and 0 otherwise.
$s(v)$ represents the locations in the image where there is a true effect.
Because the voxels are randomly selected the spheres could overlap or be cropped depending on the proximity to the edge of the brain, thus the number and size of clusters is random across simulations.
For locations where $s(v)=1$, the $d'$ variable is first orthogonalized to the other covariates and then we set
\[
\beta_4(v) = \frac{\sigma_Y(v)}{\sigma_{d'}} \times 0.4,
\]
where $\beta_4$ is as defined in \eqref{eq:simmodel} and $\sigma_{d'}$ is the variance of the orthogonalized $d'$ variable.
This approach ensures that the voxel level effect size of $d'$ is approximately equal to 0.4.
Power curves as a function of cluster size are estimated for each sample size using shape constrained additive models (SCAMs) using the {\tt scam} package in {\tt R} \citep{pya_shape_2015}.
SCAMs were used to ensure the estimated power curves were monotonically increasing.
Power was only assessed for procedures that had nearly nominal type 1 error rate.

\subsection{Results}

Simulations under the global null \eqref{eq:null} were used to assess FWER for each of the SEI procedures.
For the test of the motion variable all methods have inflated FWER for smaller sample sizes ($n=25,50$), except the PBJ and sPBJ procedures with motion dewighting (Figure \ref{fig:fwer}).
Because the sPBJ procedure without weights is still asymptotically valid the FWER falls to the nominal level with increasing sample size.
The permutation procedures and the PBJ procedure without uniform weights actually increase FWER with increasing sample size.
The GRF-based method always has a FWER above 0.9.
This pattern holds across both CFTs (0.01, 0.005) and both nominal FWER (0.01, 0.05).

Our novel null simulation approach allows us to assess how the heteroskedasticity of real variables of interest affect the FWER for each procedure.
For the test of $d'$ most procedures controlled the FWER at the nominal level (Figure \ref{fig:fwer}).
The FWER for the test of $d'$ using the PBJ procedure with motion deweighting increases above 0.1 with both CFTs when the desired level is 0.05.
The GRF-based procedure maintains a FWER above 0.4.

\begin{figure}
\center
\includegraphics[width=0.9\linewidth]{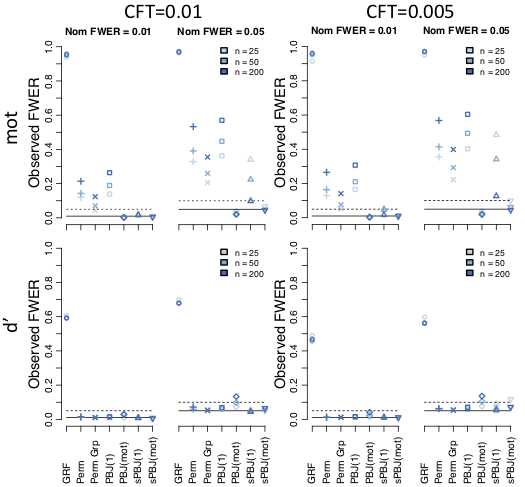}
\caption{FWER results for each of the SEI procedures. The FWER is inflated for the motion covariate except for the PBJ and sPBJ procedure with motion deweighting.
sPBJ methods approach the nominal FWER as the sample size increases.
FWERs for $d'$ are near the nominal level for all methods except the GRF procedure.
Dashed lines are 0.05 for Exp FWER$=$0.01 and 0.1 for Exp FWER$=$0.05.
CFT$=$cluster forming threshold; Nom FWER$=$nominal FWER; mot$=$motion.}
\label{fig:fwer}
\end{figure}

Simulations where known signal is added to the image are used to assess the power of each method.
The PBJ and sPBJ with motion deweighting have comparable power, whereas all methods without motion deweighting have similar, and lower, power (Figure \ref{fig:powerresults}).
Oddly, the permutation procedure with group variances has a detrimental effect on the power of the procedure.
The results demonstrate that the capability of the PBJ and sPBJ procedures to down-weight observations with high motion improves the power to detect smaller clusters of activation in small sample sizes (Figure \ref{fig:powerresults}).

\begin{figure}
\center
\includegraphics[width=0.9\linewidth]{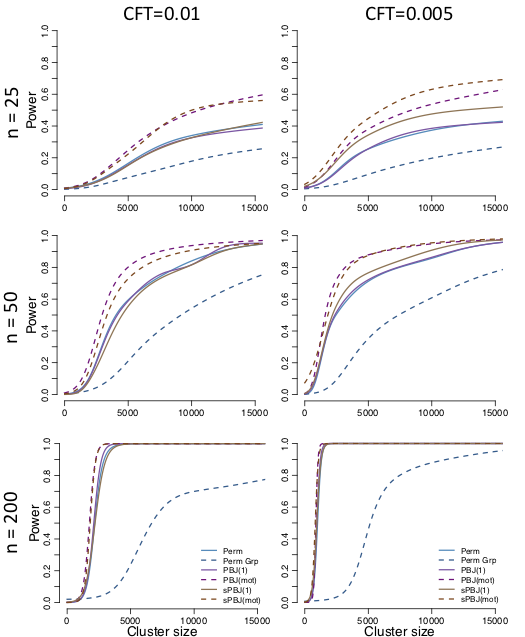}
\caption{Power results for the SEI procedures that had nearly FWER control in the null simulations. Each plot shows the power versus cluster size for cluster forming thresholds (CFTs) of 0.01 and 0.005. The PBJ and sPBJ with motion deweighting have the largest power across sample sizes. All methods without motion deweighting have lower power and the permutation method with group variance estimators has lowest power.}
\label{fig:powerresults}
\end{figure}

\section{N-back data analysis}
\label{sec:nbackanalysis}

We next investigate the effect of task performance on 2-0 back activation.
We model the 2-0 back subject-level parameter estimates using model \eqref{eq:simmodel}.
We use the inverse of the subjects' estimated variance images as weights in the estimating equation \eqref{eq:ee} for the PBJ and sPBJ procedures (using 5,000 bootstraps).
We compare the results using the (s)PBJ procedures to the permutation procedure without weights (using 5,000 permutations) with a CFT of 0.005.

The PBJ and sPBJ procedures have a larger rejection volume and a more extreme average statistical value in the rejected regions (Table \ref{tab:nbackresults}).
The permutation and sPBJ procedures identify one large contiguous cluster and one small cluster of deactivation.
The PBJ fails to detect the deactivated cluster, which is marginally significant.
Statistical maps for the PBJ and sPBJ procedure are similar as they both deweight high variance subjects.

\begin{table}
\center
\begin{tabular}{lll}
Procedure & Rej Vol (voxels) & Mean $\lvert Z \rvert$ \\
\hline
PBJ & 31,419 & 3.80 \\
sPBJ & 32,720 & 3.79 \\
Perm & 24,893 & 3.61
\end{tabular}
\caption{Rejection volume and mean $Z$-statistic in rejected voxels in the $N$-back data for each SEI procedure.}
\label{tab:nbackresults}
\end{table}

Task performance is positively associated with activation in the thalamus, dorsolateral prefrontal cortex (DLPFC), cingulate, and precuneus (Figure \ref{fig:nback}).
Performance was also associated with diffuse white matter activation that may connect functional gray matter regions or be related to reaction time differences correlated with task performance \citep{yarkoni_bold_2009,mazerolle_confirming_2010,ding_detection_2018}.
Task performance is negatively associated with activation in the ventromedial prefrontal cortex.


\begin{figure}
\center
\includegraphics[width=0.9\linewidth]{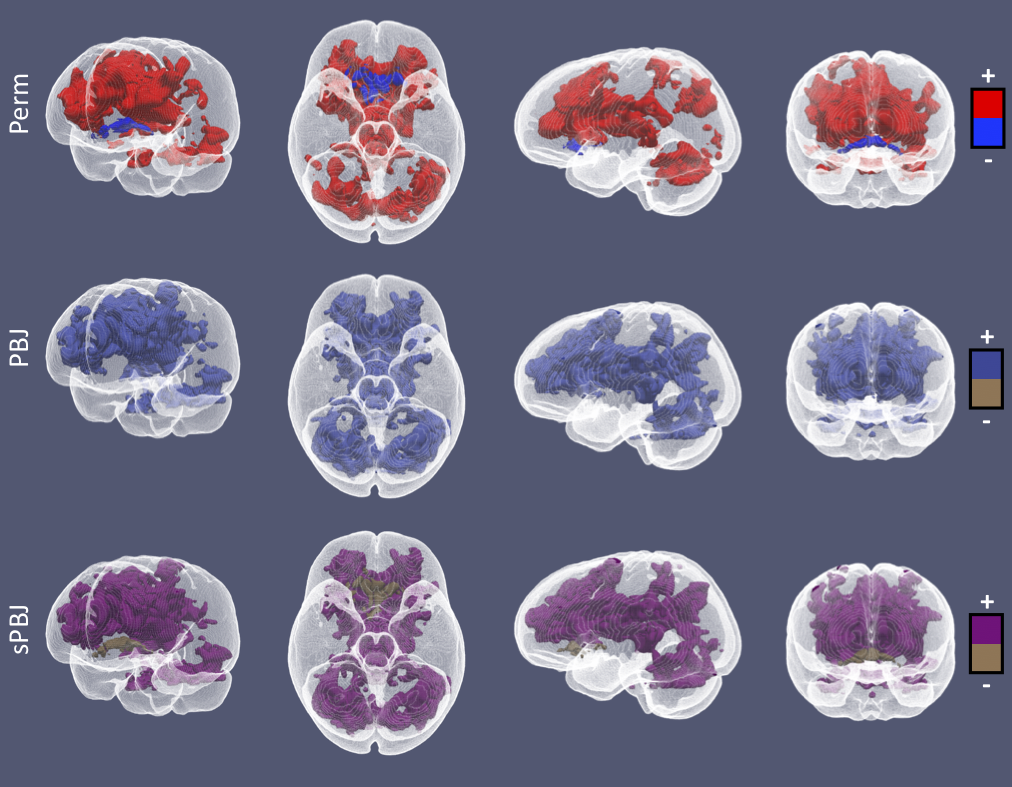}
\caption{Significant results using the three SEI procedures that controlled the FWER for the effect of $d'$ in Section \ref{sec:simulation} for a CFT of 0.005 and a spatial extent $p$-value less than 0.05. The $+$ color for each map indicates clusters that are positively activated with increasing task performance and the $-$ indicates clusters that are deactivated with increasing task performance.}
\label{fig:nback}
\end{figure}

\section{Discussion}

The sPBJ procedure is a robust multileveling approach because it produces consistent covariance function estimators even when the covariance model is misspecified.
When using the sPBJ, investigators can use arbitrary estimates of noise as weights to improve signal and still be confident that the FWER will be near nominal level in large samples.
The permutation, GRF, and PBJ procedures that do not have consistent standard error estimates can have FWERs that grow as the sample size increases.
Given that the sPBJ has closer to nominal FWER than the PBJ and permutation procedures, while maintaining equal or superior power, it is an appealing option for SEI in small and large samples.

Our findings support those presented in other papers \citep{eklund_cluster_2016,silver_false_2011,greve_false_2018}: GRF-based methods do not have robust FWER control across a range of parameter settings, regardless of whether heteroskedasticity is present.
The only variable that produces enough heteroskedasticity to affect the permutation procedure is the motion covariate; the permutation procedure performs well with $d'$, which is a typical covariate of interest.
Thus, a model that assumes subject exchangeability is appropriate for the covariate we consider here, albeit with a slight loss in power.
It is possible that covariates that are highly correlated with motion (e.g. ADHD diagnosis) may have FWERs that are far from the nominal level using the permutation procedure.

The bootstrap approach to approximating the probability \eqref{eq:pvalue} is quite general and need not be used with the two models presented here.
So long as the joint distribution of the test statistics is asymptotically normal or chi-squared and a square root of the covariance matrix can be computed, then the PBJ procedure can be used to sample from the joint distribution of the test statistics.
This suggests the method can be used to study patterns generated by machine learning methods \citep[see e.g.][]{gaonkar_analytic_2013}.

Here, we perform analyses and simulations of task fMRI data, which is one of the most common multilevel neuroimaging models.
However, the PBJ procedure is appropriate for other types of imaging modalities as subject-level weighting can be used for models that are not typically considered multilevel models.
For example, estimates of voxel level variance can be used to deweight noisy cerebral blood flow images computed from arterial spin labeled MRI or measurements of T1 image quality can be used to deweight images affected with motion artifact \citep{rosen_quantitative_2018}.
The sPBJ procedure can also be used for meta-analyses of publicly available statistical maps \citep[e.g.][]{gorgolewski_neurovault.org:_2015,maumet_sharing_2016}.
Typically, variance maps are needed to accurately fit such multilevel models, but variance maps are seldom made publicly available.
An advantage of the sPBJ procedure for meta-analysis is that it does not require variance maps for valid statistical inference, so weights could be determined from other reported variables, such as the sample size.

There are two apparent limitations of the approach:
First, the sPBJ procedure relies on consistency, which only provides guarantees for FWER control as the sample size gets large.
Our simulations demonstrate that the procedure has nominal FWER control in relatively small samples (e.g. $n=25$), however this convergence is likely slower in very heavily skewed datasets.
Second, the distribution of the statistical map \eqref{eq:wishart} is unknown when $m_1>1$, so the procedure is not able to perform tests of multiple parameters at each location.
Future work is necessary to derive the distribution of the statistical map in the case when $m_1>1$ and develop efficient methods to sample from the joint distribution of the statistical image.
Despite these limitations, the sPBJ procedure is a robust multilevel method for SEI that can be used to further our understanding of brain-behavior associations.

 \section*{Acknowledgements}
 We used BrainNet Viewer, ITK-SNAP, and Paraview to produce the figures displayed in this paper \citep{xia_brainnet_2013,yushkevich_itk-snap:_2016,ayachit_paraview_2015,madan_creating_2015}.

 \section*{Funding}
 RC2 grants from the National Institute of Mental Health MH089983 and MH089924. National Institutes of Health grants: (5P30CA068485-22 to S.N.V.); (R01MH107235 to R.C.G., R.E.G., R.T.S.), (R01MH107703 to T.D.S., R.T.S.), (R01MH112847 to TDS, KR, and RTS) and (R01NS085211 to R.T.S.).

\setlength{\bibsep}{0pt}
\begin{singlespace}
\bibliographystyle{biom}
\bibliography{./MyLibrary}
\end{singlespace}

\renewcommand{\theequation}{S\arabic{equation}}
\renewcommand{\thesection}{S\arabic{section}}

\setcounter{section}{0}
\setcounter{equation}{0}

\section{Supplementary Materials}

\subsection{Notation}
\label{sec:notation}

The following objects are defined for Theorem \ref{thm:betacovariance}.
\begin{align*}
\begin{split}
\Alpha_\beta(v) = & \E S^{-1}_i(v) X_{i1}^T X_{i1 } - \left\{\E X_{i1}^T S^{-1}_i(v) X_{i0} \right\} \left\{\E S^{-1}_i(v) X_{i0}^T X_{i0} \right\}^{-1} \left\{ \E X_{i0}^T S^{-1}_i(v) X_{i1} \right\} \\
\Omega_{\beta}(v, w) = & \left\{G(v)^T \Omega_{X_{0},X_{0}}(v, w) G(w) + \Omega_{X_{1},X_{1}}(v, w) \right.\\
& \left. - G(v)^T \Omega_{X_{0},X_{1}}(v, w) - \Omega_{X_{1},X_{0}}(v, w) G(w) \right\} \\
G(v) =  & \{\E S^{-1}_i(v) X_{i0}^T X_{i0}\}^{-1} \left\{\E S^{-1}_i(v) X_{i0}^T X_{i1}\right\} \\
\Omega_{A,B}(v, w) = &\E\left[ \{S^{-1}_i(v)S^{-1}_i(w)\}\{Y_i(v) - X_i\zeta(v)\}\{Y_i(w) - X_i\zeta(w)\} A_i^T B_i\right]
\end{split}
\end{align*}

The following values are defined for the estimator \eqref{eq:sigma} of Theorem \ref{thm:betacovariance}.
\begin{align}
S^{-1}(v) & = \diag\{[S^{-1}_1(v), \ldots, S^{-1}_n(v)]\} \in \R^{n\times n} \notag\\
P^{X_0}(v) & = \left\{ I - S^{-1/2}(v)X_0 (X_0^T S^{-1}(v) X_0)^{-1} X_0^T S^{-1/2}(v) \right\} \notag\\
\hat{A}_\beta(v) & = X_1^T S^{-1/2}(v)P^{X_0}(v)  S^{-1/2}(v) X_1 \label{eq:alpha}\\
Q_{i,i}(v) &=  P^X_{i,i}(v)^{-1}S_{i}^{-1/2}(v)\left\{ Y_i(v) - X_i\hat\zeta(v) \right\} \label{eq:Q}\\
Q_{i,j}(v) & = 0, \notag
\end{align}
and $P^X(v)$ is similarly defined to $P^{X_0}(v)$.
The use of $P^X_{i,i}(v)^{-1}$ in \eqref{eq:Q} corresponds to the heteroskedasticity-consistent estimator suggested by \citet{long_using_2000}.

\subsection{Assumptions for Theorems \ref{thm:betacovariance} and \ref{thm:m1distribution}}

The following are assumptions to ensure the covergence of the estimator \eqref{eq:sigma} in Theorem \ref{thm:betacovariance}.

\begin{equation}
\label{eq:convergenceassumptions}
\begin{aligned}
&\lim_{n\to\infty} n^{-1}\sum_{i=1}^n\text{Var}\left\{\left(X^T_{i0} S^{-1}_i(v)^{-1} X^T_{i0} \right)_{j,k} \right\}  < \infty \\
&\lim_{n\to\infty} n^{-1}\sum_{i=1}^n\text{Var}\left\{\left( X^T_{i1} S^{-1}_i(v)^{-1} X^T_{i1} \right)_{j,k} \right\} < \infty \\
&\lim_{n\to\infty} n^{-1}\sum_{i=1}^n\text{Var}\left\{\left( X^T_{i0} S^{-1}_i(v)^{-1} X^T_{i1} \right)_{j,k} \right\}  < \infty \\
&\lim_{n\to\infty} n^{-1}\sum_{i=1}^n\text{Var}\left\{Y_i(v)\right\}  < \infty.
\end{aligned}
\end{equation}

We use Theorem 18.14 of \citet{van_der_vaart_asymptotic_2000} to prove weak convergence of $T_n(v)$ which requires an additional assumption:
for all $\epsilon, \eta>0$ there exists a partition of $\mathbb{B}$ into finitely many sets $B_1, \ldots, B_K$, such that
\begin{equation}
\label{eq:vandervaart}
\lim_{n\to\infty} \P\left( \sup_i \sup{v,w\in B_i} \lvert T_n(v) - T_n(w) \rvert \right)<\eta.
\end{equation}

\subsection{Proofs of Theorems \ref{thm:betacovariance} and \ref{thm:m1distribution}}
\label{sec:proofs}

\begin{proof}{Proof of Theorem \ref{thm:betacovariance}}

Finding the covariance of $\hat\beta(v)$ and $\hat\beta(w)$ for $v,w\in B$ involves extracting the bottom right block matrix of
\begin{equation}
\Sigma_\zeta(v,w) =
\begin{bmatrix}
\Sigma_\alpha(v,w) & \Sigma_{\alpha,\beta}(v,w) \\
\Sigma_{\beta, \alpha}(v,w) & \Sigma_{\beta}(v,w)
\end{bmatrix}.
\end{equation}
Because $\Sigma_\zeta(v,w)$ is unwieldy the derivation is messy, but not intellectually stimulating and so is not shown here.

The estimator \eqref{eq:sigma} can be rewritten

\begin{align*}
& \hat \Alpha_\beta(v)^{-1} \left\{ \hat \Omega_{X_1,X_1}(v,w) - \hat G(v)^T \hat \Omega_{X_0,X_1}(v,w)  \right.\\
& \left. - \hat \Omega_{X_1,X_0}(v,w) \hat G(w) + \hat G(v)^T \hat \Omega_{X_0,X_0}(v,w) \hat G(w) \right\}\hat\Alpha_\beta(v)^{-1}
\end{align*}
where
\begin{align*}
\hat \Omega_{A,B}(v,w) & = A^T S^{-1/2}(v) Q(v) Q(w) S^{-1/2}(w) B \\
\hat G(v) & = (X^T_0 S^{-1}(v) X_0)^{-1} X_0^T S^{-1}(v) X_1,
\end{align*}
and $\hat \Alpha_\beta(v)$ is given in \eqref{eq:alpha}.

It follows that if
\begin{equation}
\label{eq:convergence1}
\begin{aligned}
n^{-1} X^T_0 S^{-1}(v) X_0 & \to_P \lim_{n\to\infty} n^{-1}\sum_{i=1}^n\E X_{i0}^T S^{-1}_i(v) X_{i0} \\
n^{-1} X^T_1 S^{-1}(v) X_1 & \to_P \lim_{n\to\infty} n^{-1}\sum_{i=1}^n\E X_{i1}^T S^{-1}_i(v) X_{i1} \\
n^{-1} X_0^T S^{-1}(v) X_1 & \to_P \lim_{n\to\infty} n^{-1}\sum_{i=1}^n\E X_{i0}^T S^{-1}_i(v) X_{i0},
\end{aligned}
\end{equation}
and
\begin{equation}
\begin{aligned}
\label{eq:convergence2}
n^{-1} X_0^T S^{-1/2}(v) Q(v) Q(w) S^{-1/2}(w) X_0 & \to_P \lim_{n\to\infty} n^{-1}\sum_{i=1}^n\E X_{i0}^T S^{-1/2}_i(v) Q(v) Q(w) S^{-1/2}_i(w)X_{i0} \\
n^{-1} X_1^T S^{-1/2}(v) Q(v) Q(w) S^{-1/2}(w) X_1 & \to_P \lim_{n\to\infty} n^{-1}\sum_{i=1}^n\E X_{i1}^T S^{-1/2}_i(v) Q(v) Q(w) S^{-1/2}_i(w)X_{i1} \\
n^{-1} X_0^T S^{-1/2}(v) Q(v) Q(w) S^{-1/2}(w) X_1 & \to_P \lim_{n\to\infty} n^{-1}\sum_{i=1}^n\E X_{i0}^T S^{-1/2}_i(v) Q(v) Q(w) S^{-1/2}_i(w)X_{i1},
\end{aligned}
\end{equation}
then $\hat \Sigma_\beta(v,w)$ is a consistent estimator of $\Sigma_\beta(v,w)$.
Equations \eqref{eq:convergence1} hold by the weak law of large numbers given assumptions \eqref{eq:convergenceassumptions}.
Note, because $Q(v)$ is a function of $\hat \zeta$ we must use a weak law that allows the averages in \eqref{eq:convergence2} to depend on parameter estimates such as Theorem 7.3 of \citep[p. 329]{boos_essential_2013}.
The assumptions \eqref{eq:convergenceassumptions} are satisfactory to use this theorem as the function
\begin{align*}
\left\lvert\frac{\partial}{\partial\zeta_j(v)}x_{i0}^T s^{-1/2}_i(v) q(v) q(w) s^{-1/2}_i(w)x_{i1} \right \rvert_\infty
 \le & ~2\left\lvert x_{i0}^Tx_{i0}^T s^{-1/2}_i(v) s^{-1/2}_i(w)x_{i1} p^X_{i,i}(v) \right \rvert_\infty y_i(v)^2 \\
& =: M(x_i, y_i, s_i, p_{i,i}),
\end{align*}
 in a neighborhood of the true parameter value $\zeta$ and $\E M(X_i, Y_i, S_i, P_{i,i}) < \infty$.
\end{proof}

\begin{proof}[Proof of Theorem \ref{thm:m1distribution}]

Let $v_1, \ldots, v_V$ be fixed for $V<\infty$.
Note that for $m_1=1$ $\Sigma_\beta(v, v) \in \R$.
The convergence of $n^{-1/2}(\hat \beta(v_1) - \beta(v_1)) \to_d N\left\{ 0, \Sigma_\beta(v_1, v_1) \right\}$ follows because our estimating equation \eqref{eq:ee} satisfies a set of regularity conditions \citep[see e.g.][]{boos_essential_2013}.
The additional assumptions \eqref{eq:AlphaOmega} and \eqref{eq:convergenceassumptions} are required to so that $Y_i$ and $S_i$ do not need to be identically distributed.
These assumptions are satisfactory for the convergence of the covariance
\[
\lim_{n\to\infty}\Cov\left\{ n^{-1/2}(\hat \beta(v_j) - \beta(v_j)), n^{-1/2}(\hat \beta(v_k) - \beta(v_k)) \right\}
= \Sigma_\beta(v_j, v_k)
\]
Thus, the continuous mapping theorem implies
\begin{align*}
Z_V = n^{-1/2}
\begin{bmatrix}
\hat\Sigma^{-1/2}_\beta(v_1, v_1) (\hat \beta(v_1) - \beta(v_1)) \\
\vdots \\
\hat\Sigma^{-1/2}_\beta(v_1, v_1) (\hat \beta(v_V) - \beta(v_V))
\end{bmatrix}
\to_d N(0, \mathrm{T})
\end{align*}
where $\mathrm{T}_{j,k} = \Sigma^{-1/2}_\beta(v_j, v_j)\Sigma_\beta(v_j, v_k)\Sigma^{-1/2}_\beta(v_k, v_k)$.
$\diag\{Z_V Z_V^T\} \sim \diag\{ \mathcal{W}(1, \mathrm{T})\}$ follows a diagonal Wishart distribution by definition.
The weak convergence of the stochastic process $T_n(v) \to_L $ follows by \eqref{eq:vandervaart} and Theorem 18.14 of \citep{van_der_vaart_asymptotic_2000}.

\end{proof}

\end{document}